\documentclass[conference,letter,romanappendices,onecolumn,draft]{ieeeconf}
\let\proof\relax   

\usepackage{amsthm,xpatch}
\usepackage{amsmath,amsfonts}
\usepackage{cite}
\usepackage{amssymb}
\usepackage{dsfont}
\usepackage{graphicx, subfigure}
\usepackage{color}
\usepackage{breqn}
\usepackage{mathtools}
\usepackage{bbm}
\usepackage{latexsym}
\usepackage[ruled, linesnumbered]{algorithm2e}
\usepackage{accents}
\usepackage{tikz}

\newtheorem{lemma}{Lemma}
\newtheorem{theorem}{Theorem}
\newtheorem{conjecture}{Conjecture}
\newtheorem{remark}{Remark}

\newcommand*{\transpose}{%
  {\mathpalette\@transpose{}}%
}

\IEEEoverridecommandlockouts

\begin{document}

\newcommand{\SB}[3]{
\sum_{#2 \in #1}\biggl|\overline{X}_{#2}\biggr| #3
\biggl|\bigcap_{#2 \notin #1}\overline{X}_{#2}\biggr|
}

\newcommand{\Mod}[1]{\ (\textup{mod}\ #1)}

\newcommand{\overbar}[1]{\mkern 0mu\overline{\mkern-0mu#1\mkern-8.5mu}\mkern 6mu}

\makeatletter
\newcommand*\nss[3]{%
  \begingroup
  \setbox0\hbox{$\m@th\scriptstyle\cramped{#2}$}%
  \setbox2\hbox{$\m@th\scriptstyle#3$}%
  \dimen@=\fontdimen8\textfont3
  \multiply\dimen@ by 4             % 4x the default rule thickness
  \advance \dimen@ by \ht0
  \advance \dimen@ by -\fontdimen17\textfont2
  \@tempdima=\fontdimen5\textfont2  % x-height
  \multiply\@tempdima by 4
  \divide  \@tempdima by 5          % 80% of the x-height
  % Modifications are only necessary if the top of the subscript is not that high:
  \ifdim\dimen@<\@tempdima
    \ht0=0pt                        % don't let the subscript interfere
    \@tempdima=\fontdimen5\textfont2
    \divide\@tempdima by 4          % 25% of the x-height
    \advance \dimen@ by -\@tempdima % if >0, add to depth of superscript!
    \ifdim\dimen@>0pt
      \@tempdima=\dp2
      \advance\@tempdima by \dimen@
      \dp2=\@tempdima
    \fi
  \fi
  #1_{\box0}^{\box2}%
  \endgroup
  }
\makeatother

\makeatletter
\renewenvironment{proof}[1][\proofname]{\par
  \pushQED{\qed}%
  \normalfont \topsep6\p@\@plus6\p@\relax
  \trivlist
  \item[\hskip\labelsep
        \itshape
%    #1\@addpunct{.}]\ignorespaces% DELETED
    #1\@addpunct{:}]\ignorespaces% ADDED
}{%
  \popQED\endtrivlist\@endpefalse
}
\makeatother

\makeatletter
\newsavebox\myboxA
\newsavebox\myboxB
\newlength\mylenA

\newcommand*\xoverline[2][0.75]{%
    \sbox{\myboxA}{$\m@th#2$}%
    \setbox\myboxB\null% Phantom box
    \ht\myboxB=\ht\myboxA%
    \dp\myboxB=\dp\myboxA%
    \wd\myboxB=#1\wd\myboxA% Scale phantom
    \sbox\myboxB{$\m@th\overline{\copy\myboxB}$}%  Overlined phantom
    \setlength\mylenA{\the\wd\myboxA}%   calc width diff
    \addtolength\mylenA{-\the\wd\myboxB}%
    \ifdim\wd\myboxB<\wd\myboxA%
       \rlap{\hskip 0.5\mylenA\usebox\myboxB}{\usebox\myboxA}%
    \else
        \hskip -0.5\mylenA\rlap{\usebox\myboxA}{\hskip 0.5\mylenA\usebox\myboxB}%
    \fi}
\makeatother

\xpatchcmd{\proof}{\hskip\labelsep}{\hskip3.75\labelsep}{}{}

\pagestyle{plain}

\title{\fontsize{21}{28}\selectfont Single-Server Multi-Message Individually-Private Information Retrieval with Side Information}

\author{Anoosheh Heidarzadeh, Swanand Kadhe, Salim El Rouayheb, and Alex Sprintson\thanks{A.~Heidarzadeh and A.~Sprintson are with the Department of Electrical and Computer Engineering, Texas A\&M University, College Station, TX 77843 USA (E-mail: \{anoosheh,spalex\}@tamu.edu).}\thanks{S.~Kadhe is with the Department of Electrical Engineering and Computer Sciences, University of California, Berkeley, CA 94720 USA (E-mail: swanand.kadhe@berkeley.edu).}\thanks{S.~El Rouayheb is with the Department of Electrical and Computer Engineering, Rutgers University, Piscataway, NJ 08854 USA (E-mail: sye8@soe.rutgers.edu).}}

%\thanks{This work was supported by the National Science Foundation under Grant No.~CNS-0954153 and the AFOSR under Contract No.~FA9550-13-1-0008.}

\maketitle 

\thispagestyle{plain}

\begin{abstract} 
We consider a multi-user variant of the private information retrieval problem described as follows. Suppose there are $D$ users, each of which wants to privately retrieve a distinct message from a server with the help of a \emph{trusted agent}. We assume that the agent has a random subset of $M$ messages that is not known to the server. The goal of the agent is to collectively retrieve the users' requests from the server. For protecting the privacy of users, we introduce the notion of \emph{individual-privacy} -- the agent is required to protect the privacy only for each individual user (but may leak some correlations among user requests). We refer to this problem as \emph{Individually-Private Information Retrieval with Side Information (IPIR-SI)}.

We first establish a lower bound on the capacity, which is defined as the maximum achievable download rate, of the IPIR-SI problem by presenting a novel achievability protocol. Next, we characterize the capacity of IPIR-SI problem for $M = 1$ and ${D = 2}$. In the process of characterizing the capacity for arbitrary $M$ and $D$ we present a novel combinatorial conjecture, that may be of independent interest.
\end{abstract}

\section{introduction}
\label{sec:intro}
%{\color{red} Swanand, can you please focus on information-theoretic PIR references only?}
In the conventional {Private Information Retrieval (PIR) problem}, a user wants to privately download a message belonging to a database with copies stored on a single or multiple remote servers (see~\cite{yekhanin2010private}). 
The multiple-server PIR problem has been predominantly studied in the PIR literature, with  breakthrough results for the information-theoretic privacy model in the past few years (see e.g., \cite{Sun2017, JafarPIR3new, tajeddine2017private1,BU18}, and references therein).
The multi-message extension of the PIR problem enables a user to privately download multiple messages from the server(s)~\cite{Banawan2017,Maddah2018}. 
There have been a number of recent works on the PIR problem when some side information is present at the user~\cite{Tandon2017,Wei2017CacheAidedPI,Wei2017FundamentalLO,Chen2017side,Maddah2018}.

% The multiple server case of the PIR problem has received extensive attention in the literature (see, e.g., \cite{beimel2002breaking,gasarch2004survey, yekhanin2010private,Sun2017,JafarPIR3new,Banawan2017}). 
% However, the model and the information-theoretic privacy guarantees for the multi-server case rely on the assumption of restricted collusion among the servers. This assumption is hard to verify and enforce, and can also be infeasible in some practical systems; for instance when the servers are owned and operated by same entity. Therefore, we focus our attention on the practically appealing single-server PIR model, wherein the no-collusion assumption is no longer {relevant}.

% One downside of the single-server PIR model is that if the user desires information-theoretic privacy, it is required to download the entire database. Recently, it is shown in~\cite{Kadhe2017} that, if the user knows a random subset of messages that is unknown to the server, they can substantially reduce the download cost and still achieve information-theoretic privacy for the requested message. 

Recently, in~\cite{Kadhe2017,KGHERS2017}, the authors considered the single-server PIR with Side Information (PIR-SI) problem, wherein the user knows a random subset of messages that is unknown to the server. %This side information could have been obtained from other trusted users or from previous interactions with the server.
It was shown that the side information enables the user to substantially reduce the download cost and still achieve information-theoretic privacy for the requested message. The multi-message version of PIR-SI is considered in~\cite{HKGRS:2018,LG:2018}, and the case of coded side information is considered in~\cite{HKS2018}. Single-server multi-user PIR-SI problem wherein all users have the same demand but different side-information sets was considered in~\cite{LiG:18:MU}.

In this work, we consider the following scenario. Suppose there are $D$ users, each of which wants to privately retrieve a distinct message from a server. The users send their demands to a {\it trusted agent}. The agent has a subset of $M$ messages, unknown to the server. This side information could have been obtained from the users or from previous interactions with the server. Followed by aggregating the users' requests, the agent then collectively retrieves information from the server. 

One natural solution for the agent to achieve privacy during the retrieval is to successively use the PIR-SI protocol in~\cite{Kadhe2017} for each request. However, the agent can achieve much higher download rate while preserving the privacy collectively for all the users by using the multi-message PIR protocol in~\cite{HKGRS:2018,LG:2018}. In this work, we introduce the notion of \emph{individual-privacy} where the agent is required to protect the privacy only for each individual user, and we refer to this problem as \emph{Individually-Private Information Retrieval with Side Information (IPIR-SI)}. We seek to answer the following questions: is it possible to further increase the download rate when individual-privacy is required? Moreover, what are the fundamental limits on the download rate for the IPIR-SI problem? We answer the first question affirmatively and take the first steps towards answering the second question. %Our contributions can be summarized as follows. 

\subsection{Main Contributions}
%Focusing our attention to information-theoretic privacy, we consider the following multi-user variant of the single-server PIR-SI problem. There are $D$ users, each of which demands one distinct message from the server. The users communicate their requests to a trusted agent, who possesses $M$ messages such that their identities are unknown to the server. The agent aims to collectively download all the requests while protecting the privacy of users. In particular, we propose the notion of {\it Individual Privacy} (IP) that enables the agent to preserve the privacy for individual users. Our goal is to characterize the capacity of this setting, i.e., the maximum achievable download rate over all such IPIR-SI schemes. 

%Towards this, 
We first establish a lower bound on the capacity of the IPIR-SI problem (where the capacity is defined as the supremum of all achievable download rates) by presenting a new protocol which builds up on the Generalized Partition and Code (GPC) protocol in~\cite{HKGRS:2018}. Next, we characterize the capacity of IPIR-SI problem for $M = 1$ and ${D = 2}$. In the process of characterizing the capacity for arbitrary $M$ and $D$ we present a novel combinatorial conjecture, that may be of independent interest. 

For $M=1$ and arbitrary $D$, our conjecture relates the size of an external mother vertex-set (i.e., a minimal subset of nodes from which any other node with nonzero our-degree can be reached via a directed path) of any directed graph $G$ with certain bounds on the in-degree and out-degree of the nodes, to the size of an internal mother vertex-set (i.e., a minimal subset of nodes from which any other node with nonzero in-degree can be reached via a directed path) of the transpose of $G$ which is obtained by reversing the direction of all edges in $G$. %{\color{blue} Anoosheh, can you please elaborate? Just saying that a conjecture in combinatorics is a bit vague. I could try to write about mother vertex-set, but is that our main conjecture. I couldn't to write something concrete about Conjecture 1.}

% \subsection{Related Work}
% {\color{red} Swanand, can you please summarize the works on the information-theoretic PIR in the first paragraph of intro? About 15 references in total. I don't think we need to have a separate Related Work subsection here.}
% The replication-based model where multiple servers store copies of the database has been predominantly studied in the PIR literature, with  breakthrough results in the past few years (e.g., \cite{Sun2017, JafarPIR3new, yekhanin2010private, beimel2001information, beimel2002breaking,gasarch2004survey}). The multi-message PIR problem for multiple non-colluding servers was recently considered in~\cite{Banawan2017} under information-theoretic privacy. 
% There have been a number of recent works on the capacity of PIR when some side information is present at the user for single server~\cite{Kadhe2017,HKS2018,HKGRS:2018,LG:2018} and multiple servers~\cite{Tandon2017,Wei2017CacheAidedPI,Wei2017FundamentalLO,KGHERS2017,Chen2017side,Maddah2018}.
% Single-server multi-user PIR-SI problem wherein all users have the same demand but different side-information sets was considered in~\cite{LiG:18:MU}. Another related line of work is that of {private broadcasting} in~\cite{Karmoose2017,Narayanan:18}, which considers the index coding setting with multiple users with side information and a single server. The privacy constraint  is to protect the request and side information of a user from the other users. 

\section{Problem Formulation}\label{sec:SN}
Let $\mathbb{F}_q$ be a finite field of size $q$, and let $\mathbb{F}_{q^m}$ be an extension field of $\mathbb{F}_q$ for some integer $m\geq 1$. Let $L \triangleq m\log_2 q$, and let ${\mathbb{F}_q^{\times} \triangleq \mathbb{F}_q\setminus \{0\}}$. For a positive integer $i$, we denote $\{1,\dots,i\}$ by $[i]$. Also, let $K\geq 1$, $M\geq 1$, and $D\geq 1$ be arbitrary integers such that $D+M\leq K$. 

%Suppose that there is a server storing a set $X$ of $K$ messages $X_1,\dots,X_K$, with each message $X_i$ being independently and uniformly distributed over $\mathbb{F}_{q^m}$, i.e., ${H(X_1) = \dots = H(X_K) = L}$ and $H(X_1,\dots,X_K) = KL$, where $L \triangleq m\log_2 q$. Also, suppose that there is a user that wishes to retrieve $D$ messages $X_W\triangleq \{X_{j_1},\dots,X_{j_D}\}$ from the server for some $W\triangleq \{j_1,\dots,j_D\} \in \mathcal{W}$, and knows $M$ messages $X_S\triangleq \{X_{i_1},\dots,X_{i_M}\}$ for some $S \triangleq \{i_1,\dots,i_M\}\in \mathcal{S}$ such that $S\cap W=\emptyset$. We refer to $W$ as the \emph{demand index set}, $X_W$ as the \emph{demand}, $D$ as the \emph{demand size}, $S$ as the \emph{side information index set}, $X_S$ as the \emph{side information}, and $M$ as the \emph{side information size}. 

Suppose that there is a server storing a set of $K$ messages $X_1,\dots,X_K$, with each message $X_i$ being independently and uniformly distributed over $\mathbb{F}_{q^m}$, i.e., ${H(X_1) = \dots = H(X_K) = L}$ and $H(X_1,\dots,X_K) = KL$. Also, suppose that there are $D$ users, each of which demands one distinct message $X_j$. Let $W$ be the index set of the users' demanded messages. The users send the indices of their demanded messages to a trusted agent, called \emph{aggregator}, who knows $M$ messages $X_S\triangleq \{X_j\}_{j\in S}$ for some $S\subset [K]$, $|S|=M$, $S\cap W=\emptyset$. Then, the aggregator retrieves the $D$ messages $X_W\triangleq \{X_j\}_{j\in W}$ from the server. We refer to $W$ as the \emph{demand index set}, $X_W$ as the \emph{demand}, $D$ as the \emph{demand size}, $S$ as the \emph{side information index set}, $X_S$ as the \emph{side information}, and $M$ as the \emph{side information size}. 

Denote by $\mathcal{W}$ and $\mathcal{S}$ the set of all subsets of $\mathcal{K}\triangleq [K]$ of size $D$ and $M$, respectively. Also, let $\boldsymbol{S}$ and $\boldsymbol{W}$ be two random variables representing $S$ and $W$, respectively. Denote the probability mass function (PMF) of $\boldsymbol{S}$ by $p_{\boldsymbol{S}}(\cdot)$ and the conditional PMF of $\boldsymbol{W}$ given $\boldsymbol{S}$ by $p_{\boldsymbol{W}|\boldsymbol{S}}(\cdot|\cdot)$. We assume that $\boldsymbol{S}$ is uniformly distributed over $\mathcal{S}$, i.e., $p_{\boldsymbol{S}}(S) = \binom{K}{M}^{-1}$ for all $S\in \mathcal{S}$, and $\boldsymbol{W}$ (given $\boldsymbol{S}=S$) is uniformly distributed over ${\{W\in \mathcal{W}:W\cap S=\emptyset\}}$, i.e., 
\begin{equation*}
p_{\boldsymbol{W}|\boldsymbol{S}}(W|S) = 
\left\{\begin{array}{ll}
\binom{K-M}{D}^{-1}, & W\in \mathcal{W}, W\cap S=\emptyset,\\	
0, & \text{otherwise}.
\end{array}\right.	
\end{equation*} %Note that $\mathbb{P}(j\in \boldsymbol{W})=\frac{D}{K}$ for all $j\in \mathcal{K}$.
 
We assume that the server knows the size of $W$ (i.e., $D$) and the size of $S$ (i.e., $M$), as well as the PMF $p_{\boldsymbol{S}}(\cdot)$ and the conditional PMF $p_{\boldsymbol{W}|\boldsymbol{S}}(\cdot|\cdot)$, whereas the realizations $S$ and $W$ are unknown to the server a priori. 

For any $S$ and $W$, in order to retrieve $X_W$, the aggregator sends to the server a query $Q^{[W,S]}$, which is a (potentially stochastic) function of $W$, $S$, and $X_S$. The query $Q^{[W,S]}$ must protect from the privacy of the demand index of every user individually from the server, i.e., \[\mathbb{P}(j\in \boldsymbol{W}| Q^{[W,S]})= \mathbb{P}(j\in \boldsymbol{W})=\frac{D}{K}\] for all $j \in \mathcal{K}$. We refer to this condition as the \emph{individual-privacy condition}. Note that the individual-privacy condition is weaker than the \emph{joint-privacy condition}, also known as the \emph{$W$-privacy condition}, being studied in~\cite{HKGRS:2018}, where the privacy of all indices in the demand index set must be protected jointly. %, i.e., ${\mathbb{P}(\boldsymbol{W}=W'|Q^{[W,S]})=\mathbb{P}(\boldsymbol{W}=W')}$ for all $W'\in \mathcal{W}$. 
The notions of individual privacy and joint privacy coincide for $D=1$, which was previously settled in~\cite{Kadhe2017}, and hence, in this work, we focus on $D\geq 2$.

Upon receiving $Q^{[W,S]}$, the server sends to the aggregator an answer $A^{[W,S]}$, which is a (deterministic) function of the query $Q^{[W,S]}$ and the messages in $X$, i.e., \[H(A^{[W,S]}| Q^{[W,S]},\{X_j\}_{j\in \mathcal{K}}) = 0.\] The answer $A^{[W,S]}$ along with the side information $X_S$ must enable the aggregator to retrieve the demand $X_W$, i.e., \[H(X_W| A^{[W,S]}, Q^{[W,S]}, X_S)=0.\] This condition is referred to as the \emph{recoverability condition}. 

The problem is to design a query $Q^{[W,S]}$ and an answer $A^{[W,S]}$ (for any $W$ and $S$) that satisfy the individual-privacy and recoverability conditions. We refer to this problem as \emph{single-server multi-message Individually-Private Information Retrieval with Side Information (IPIR-SI)}. 

A collection of $Q^{[W,S]}$ and $A^{[W,S]}$ (for all $W$ and $S$) which satisfy the individual-privacy and recoverability conditions, is referred to as an \emph{IPIR-SI protocol}. We define the \emph{rate} of an IPIR-SI protocol as the ratio of the entropy of the demand messages, i.e., $DL$, to the average entropy of the answer, i.e., $H(A^{[\boldsymbol{W},\boldsymbol{S}]})=\sum H(A^{[W,S]})p_{\boldsymbol{W}|\boldsymbol{S}}(W|S)p_{\boldsymbol{S}}(S)$, where the average is taken over all $W$ and $S$. The \emph{capacity} of the IPIR-SI problem is also defined as the supremum of rates over all IPIR-SI protocols. 

In this work, our goal is to characterize the capacity of the IPIR-SI problem, and to design an IPIR-SI protocol that achieves the capacity. 

% , {\color{blue} which is tight for the special case of $M=1$ and $D=2$,}

\section{Main Results}
In this section, we present our main results. Theorem~\ref{thm:IPIRSI-Gen} provides a lower bound on the capacity of IPIR-SI problem for $M\geq 1$ and $D\geq 2$, and Theorem~\ref{thm:IPIRSI-Spc} characterizes the capacity of IPIR-SI problem for the special case of $M=1$ and $D=2$. The proofs of Theorems~\ref{thm:IPIRSI-Gen} and~\ref{thm:IPIRSI-Spc} are given in Sections~\ref{sec:IPIRSI-Gen} and~\ref{sec:IPIRSI-Spc}, respectively.  

\begin{theorem}\label{thm:IPIRSI-Gen}
The capacity of IPIR-SI problem with $K$ messages, side information size $M\geq 1$, and demand size $D\geq 2$ is lower bounded by $D(K-M\lfloor\frac{K}{M+D}\rfloor)^{-1}$ if $\frac{K-D}{M+D}\leq \lfloor \frac{K}{M+D}\rfloor$, and by $\lceil \frac{K}{M+D}\rceil^{-1}$ otherwise.
\end{theorem}

The proof is based on constructing an IPIR-SI protocol that achieves the rate ${D(K-M\lfloor K/(M+D)\rfloor)^{-1}}$ or ${\lceil K/(M+D)\rceil^{-1}}$, depending on $K$, $M$, and $D$ (see, for details, Section~\ref{sec:IPIRSI-Gen}). This protocol, which is a variation of the Generalized Partition and Code (GPC) protocol previously proposed in~\cite{HKGRS:2018} for single-server multi-message PIR-SI where joint-privacy is required, is referred to as \emph{GPC for Individual Privacy}, or \emph{GPC-IP} for short.

\begin{remark}
\emph{A lower bound on the capacity of single-server multi-message PIR with side information, when the privacy of the demand indices must be protected jointly, was previously presented in~\cite[Theorem~1]{HKGRS:2018}. Surprisingly, this lower bound reduces to the lower bound of Theorem~\ref{thm:IPIRSI-Gen} where $M$ (in ~\cite[Theorem~1]{HKGRS:2018}) is replaced by $MD$. This correspondence implies that each message in the side information, when achieving individual-privacy, can be as effective as $D$ side information messages when joint-privacy is required. This also suggests that, as one would expect, relaxing the privacy condition (from joint to individual) can increase the capacity.}  
\end{remark}

\begin{theorem}\label{thm:IPIRSI-Spc}
The capacity lower bound given in Theorem~\ref{thm:IPIRSI-Gen} is tight for $M=1$ and $D=2$.
\end{theorem}

% For the case of $M=0$, the proof of converse follows from a simple information-theoretic argument, and for the case of $D=1$, the converse proof is a direct result of~\cite[Theorem~]{}. 

The proof of converse is based on a mixture of new combinatorial and information-theoretic arguments relying on two necessary conditions imposed by the individual-privacy and recoverability conditions (see Lemmas~\ref{lem:2} and~\ref{lem:3}). %For arbitrary $M$ and $D$, the tightness of this result, however, remains open in general.

% The achievability proof for both cases relies on the GPC-IP protocol, which reduces to the obvious protocol of downloading all messages for $M=0$ and to the Partition and Code protocol of~\cite{} for $D=1$.

\begin{remark}
\emph{As we will show later, the tightness of the result of Theorem~\ref{thm:IPIRSI-Gen} for arbitrary $M$ and $D$, which remains open in general, is conditional upon the correctness of a novel conjecture in combinatorics, formally stated in Section~\ref{sec:IPIRSI-Spc}, which may be of independent interest. Interestingly, for ${M=1}$ and ${D\geq 2}$, our conjecture relates the size of an external mother vertex-set of any directed graph $G$, whose nodes have in-degree at least one and out-degree either zero or at least $D$, to the size of an internal mother vertex-set of the transpose of $G$ (which is the graph obtained by reversing the direction of all edges in $G$). (The notions of external and internal mother vertex-sets, formally defined in Section~\ref{sec:IPIRSI-Spc}, are two generalizations of the notion of the mother vertex in graph theory.) In this work, we prove the simplest non-trivial case of this conjecture for $M=1$ and $D=2$, and leave the complete proof for the future work.}	
\end{remark}

\section{Proof of Theorem~\ref{thm:IPIRSI-Gen}}\label{sec:IPIRSI-Gen}

In this section, we propose an IPIR-SI protocol, referred to as \emph{Generalized Partition and Code for Individual Privacy (GPC-IP)}, achieving the rate lower bound of Theorem~\ref{thm:IPIRSI-Gen}. 

Define $\alpha\triangleq M+D$, $\beta\triangleq \lfloor K/\alpha\rfloor$, and $\rho\triangleq K-\alpha\beta$. (Note that $0\leq \rho<\alpha$.) Also, define $\gamma\triangleq \min\{\rho,D\}$. Assume that $q\geq \alpha$, and let $\omega_1,\dots,\omega_{\alpha}$ be $\alpha$ distinct elements from $\mathbb{F}_q$. 

\textit{\bf GPC-IP Protocol:} This protocol consists of four steps as follows: 

\textbf{\it Step 1:} First, the aggregator constructs a set $Q_{0}$ of size $\rho$ from the indices in $\mathcal{K}$, and $\beta$ disjoint sets $Q_1,\dots,Q_{\beta}$ (also disjoint from $Q_0$), each of size $\alpha$, from the indices in $\mathcal{K}$, where the construction procedure is described below. 

Define \[\theta_1\triangleq \frac{\binom{\alpha-1}{M}}{\prod_{i=1}^{\beta-1}\binom{K-i\alpha}{\alpha}},\] \[\theta_2 \triangleq \frac{\binom{\alpha-1}{M+\rho}\binom{M+\rho}{M}(\frac{\alpha\beta}{D-\rho}-1)}{\binom{D}{\rho}\binom{K-\alpha}{\rho}\prod_{i=1}^{\beta-1}\binom{K-i\alpha-\rho}{\alpha}},\] and \[\theta_3\triangleq \frac{\beta\binom{\rho}{D}\binom{K-\rho}{\alpha-\rho}}{\binom{M}{\rho-D}\prod_{i=0}^{\beta-1}\binom{K-i\alpha-\rho}{\alpha}}.\] 

There are two cases based on $\rho$: (i) ${\rho<D}$, and (ii) ${\rho\geq D}$. 

\emph{Case (i):} With probability $\frac{\theta_1}{\theta_1+\theta_2}$, the aggregator places $\rho$ randomly chosen elements (demand indices) from $W$ into $Q_0$ and the remaining elements in $W$ along with all elements in $S$ (side information indices) into $Q_1$. Then the aggregator randomly places all other elements in $\mathcal{K}$ into $Q_2,\dots,Q_{\beta}$ and the remaining positions in $Q_1$; otherwise, with probability $\frac{\theta_2}{\theta_1+\theta_2}$, the aggregator places all elements in $S\cup W$ into $Q_1$, and randomly places all other elements in $\mathcal{K}$ into $Q_0,Q_2,\dots,Q_{\beta}$. 

\emph{Case (ii):} With probability $\frac{\theta_1}{\theta_1+\theta_3}$, the aggregator places all elements in $W$ along with $\rho-D$ randomly chosen elements from $S$ into $Q_0$, and places the remaining elements of $S$ together with all other elements in $\mathcal{K}$ into $Q_1,\dots,Q_{\beta}$ at random; otherwise, with probability $\frac{\theta_3}{\theta_1+\theta_3}$, the aggregator places all elements in $S\cup W$ into $Q_1$, and randomly places all other elements in $\mathcal{K}$ into $Q_0,Q_2,\dots,Q_{\beta}$.

Next, the aggregator creates a collection $Q'$ of $\gamma$ sequences $Q'_{1},\dots,Q'_{\gamma}$, each of length $\rho$, such that $Q'_{i}  = \{\omega^{i-1}_1,\dots,\omega^{i-1}_{\rho}\}$ for $i\in [\gamma]$, and a collection $Q''$ of $D$ sequences $Q''_{1},\dots,Q''_{D}$, each of length $\alpha$, such that $Q''_{i}  = \{\omega^{i-1}_1,\dots,\omega^{i-1}_{\alpha}\}$ for $i\in [D]$.

\textbf{\it Step 2:} The aggregator constructs ${Q^{*}_0 = (Q_0,Q')}$ and ${Q^{*}_i = (Q_i,Q'')}$ for $i\in [\beta]$, and sends to the server the query ${Q^{[W,S]} = \{Q^{*}_{0},Q^{*}_{\sigma^{-1}(1)},\dots,Q^{*}_{\sigma^{-1}(\beta)}\}}$ for a randomly chosen permutation $\sigma: [\beta]\rightarrow [\beta]$.

\textbf{\it Step 3:} By using ${Q^{*}_0 = (Q_0,Q')}$ and ${Q^{*}_i = (Q_i,Q'')}$ for $i\in [\beta]$, the server computes ${A_0 = \{A^1_0,\dots,A^{\gamma}_0\}}$ by ${A^j_0 = \sum_{l=1}^{\rho} \omega_l^{j-1} X_{i_l}}$ for ${j\in [\gamma]}$ where ${Q_0 = \{i_1,\dots,i_{\rho}\}}$, and computes ${A_i = \{A^1_i,\dots,A^D_i\}}$ for ${i\in [\beta]}$ by ${A^{j}_{i} = \sum_{l=1}^{\alpha} \omega^{j-1}_{l} X_{i_j}}$ for $j\in [D]$ where ${Q_i = \{i_1,\dots,i_{\alpha}\}}$. The server then sends to the aggregator the answer ${A^{[W,S]}=\{A_{0},A_{\sigma^{-1}(1)}\dots,A_{\sigma^{-1}(\beta)}\}}$.

\textbf{\it Step 4:} Upon receiving the answer from the server, the aggregator retrieves $X_j$ for any $j\in W\cap Q_0$ (or any $j\in W\cap Q_i$ for some $i\in [\beta]$) by subtracting off the contribution of the side information messages $\{X_i\}_{i\in S}$ from the $\gamma$ (or $D$) equations in $A_{0}$ (or $A_i$), and solving the resulting system of $\gamma$ (or $D$) linear equations with $\gamma$ (or $D$) unknowns.

\begin{lemma}\label{lem:Ach}
The GPC-IP protocol is an IPIR-SI protocol, and achieves the rate ${D(K-M\lfloor \frac{K}{M+D}\rfloor)^{-1}}$ if $\frac{K-D}{M+D}\leq \lfloor \frac{K}{M+D}\rfloor$, and the rate ${\lceil \frac{K}{M+D}\rceil^{-1}}$ otherwise.
\end{lemma}

\begin{proof}
If $\frac{K-D}{M+D}\leq \lfloor \frac{K}{M+D}\rfloor$, then $\rho< D$. Thus, $\gamma = \rho$. In this case, $H(A_0) = \rho L$ and $H(A_i) = DL$ for $i\in [\beta]$, where $L=H(X_i)$ for all $i$. Thus, for any $W\in \mathcal{W}$ and $S\in \mathcal{S}$ such that $S\cap W=\emptyset$, we have $H(A^{[W,S]})=H(A_0,\dots,A_{\gamma})=\sum_{i=0}^{\beta} H(A_i)=(\rho+\beta D)L$. Thus, in this case, the rate is $DL/H(A^{[\boldsymbol{W},\boldsymbol{S}]}) = DL/H(A^{[W,S]})=D/(\rho+\beta D)$, or equivalently, $D(K-M\lfloor\frac{K}{M+D}\rfloor)^{-1}$. If $\frac{K-D}{M+D}> \lfloor \frac{K}{M+D}\rfloor$, then $\rho\geq D$. Thus, ${\gamma = D}$. In this case, $H(A_0) = H(A_i) = DL$ for $i\in [\beta]$, and thus, ${H(A^{[W,S]})=(\beta+1)DL}$. In this case, the rate is $D/(\beta+1)D$, or equivalently, ${\lceil \frac{K}{M+D}\rceil^{-1}}$. 

Next, we prove that the GPC-IP protocol is an IPIR-SI protocol. It should be easy to see that the recoverability condition is satisfied. We only need to prove that the GPC-IP protocol satisfies the individual-privacy condition. 
% (for definition of $Q_i$'s, see Step~1 of the GPC-IP protocol)

Consider an arbitrary ${Q \triangleq \{Q_0,\dots,Q_{\beta}\}}$. We need to show that ${\mathbb{P}(j\in \boldsymbol{W}|Q)=\mathbb{P}(j\in \boldsymbol{W})}$ for all $j\in \mathcal{K}$. Equivalently, it suffices to show $\mathbb{P}(j\in \boldsymbol{W}|Q)$ is the same for all $j\in \mathcal{K}$. %Due to the lack of space, we only give the proof for the case of $\rho<D$, and refer the reader to the long version of this work~\cite{HKRS2019} for the proof of the case of $\rho\geq D$. 

First, suppose that $\rho<D$. It is easy to see that $\mathbb{P}(j\in \boldsymbol{W}|Q)$ is given by
\begin{dmath}\label{Eq:1}
\sum_{i=1}^{\beta} \sum_{\substack{W\subset Q_i: \\ |W|=D-\rho}} \sum_{\substack{S\subset Q_i\setminus W:\\ |S|=M}}\mathbb{P}(\boldsymbol{W}=Q_0\cup W, \boldsymbol{S}=S|Q)  
\end{dmath} for all $j\in Q_0$, and 
\begin{dmath}\label{Eq:2}
 \sum_{\substack{W\subset Q_i:\\ |W|=D, j\in W}} \mathbb{P}(\boldsymbol{W}=W, \boldsymbol{S}=Q_i\setminus W|Q) + {\sum_{\substack{W\subset Q_i:\\ |W|=D-\rho, j\in W}} \sum_{\substack{S\subset Q_i\setminus W:\\ |S|=M}} \mathbb{P}(\boldsymbol{W}=Q_0\cup W, \boldsymbol{S}=S|Q)}  
\end{dmath} for all $j\in Q_i$, $i\in [\beta]$. From~\eqref{Eq:1} and~\eqref{Eq:2}, one can see that $\mathbb{P}(j\in \boldsymbol{W}|Q)$ is the same for all $j\in Q_0$, say equal to $p_0$, and is the same for all $j\in Q_i$ and all $i\in [\beta]$, say equal to $p_1$. We need to show that $p_0$ and $p_1$ are equal. It is easy to show that $p_0$ and $p_1$ are equal if the two quantities
\begin{dmath}\label{Eq:3}
\sum_{i=1}^{\beta} \sum_{\substack{W\subset Q_i: \\ |W|=D-\rho}} \sum_{\substack{S\subset Q_i\setminus W:\\ |S|=M}} \mathbb{P}(Q|\boldsymbol{W}=Q_0\cup W, \boldsymbol{S}=S) 
\end{dmath} and 
\begin{dmath}\label{Eq:4}
 \sum_{\substack{W\subset Q_i:\\ |W|=D, j\in W}} \mathbb{P}(Q|\boldsymbol{W}=W, \boldsymbol{S}=Q_i\setminus W) + {\sum_{\substack{W\subset Q_i:\\ |W|=D-\rho, j\in W}} \sum_{\substack{S\subset Q_i\setminus W:\\ |S|=M}} \mathbb{P}(Q|\boldsymbol{W}=Q_0\cup W, \boldsymbol{S}=S)}  
\end{dmath} are equal. Fix an $i\in [\beta]$. For any $W\subset Q_i$, $|W|=D-\rho$, and any $S\subset Q_i\setminus W$, $|S|=M$, a simple counting yields
\begin{align*}
& \mathbb{P}(Q|\boldsymbol{W}=Q_0\cup W, \boldsymbol{S}=S) = \left(\frac{\theta_1}{\theta_1+\theta_2}\right) (\beta-1)!\left(\binom{D}{\rho}\binom{K-\alpha}{\rho}\prod_{i=1}^{\beta-1}\binom{K-i\alpha-\rho}{\alpha}\right)^{-1}.	
\end{align*} and accordingly,~\eqref{Eq:3} is equal to 
\begin{align*}
& \left(\frac{\theta_1}{\theta_1+\theta_2}\right) \binom{\alpha}{M+\rho}\binom{M+\rho}{M} \beta!\left(\binom{D}{\rho}\binom{K-\alpha}{\rho}\prod_{i=1}^{\beta-1}\binom{K-i\alpha-\rho}{\alpha}\right)^{-1}.	
\end{align*} For any $W\subset Q_i$, $|W|=D$ such that $j\in W$, we have 
\begin{align*}
& \mathbb{P}(Q|\boldsymbol{W}=W, \boldsymbol{S}=Q_i\setminus W) = \left(\frac{\theta_2}{\theta_1+\theta_2}\right) (\beta-1)!\left(\prod_{i=1}^{\beta-1}\binom{K-i\alpha}{\alpha}\right)^{-1}, 
\end{align*} and for any $W\subset Q_i$, $|W|=D-\rho$ such that $j\in W$, and any $S\subset Q_i\setminus W$, $|S|=M$, we have 
\begin{align*}
& \mathbb{P}(Q|\boldsymbol{W}=Q_0 \cup W, \boldsymbol{S}=S) = \left(\frac{\theta_1}{\theta_1+\theta_2}\right) (\beta-1)!\left(\binom{D}{\rho}\binom{K-\alpha}{\rho}\prod_{i=1}^{\beta-1}\binom{K-i\alpha-\rho}{\alpha}\right)^{-1}. 
\end{align*} Accordingly,~\eqref{Eq:4} is equal to 
\begin{align*}
& \binom{\alpha-1}{M} (\beta-1)! \left(\left(\frac{\theta_2}{\theta_1+\theta_2}\right)\left(\prod_{i=1}^{\beta-1}\binom{K-i\alpha}{\alpha}\right)^{-1}\right.\\
& \quad + \left(\frac{\theta_1}{\theta_1+\theta_2}\right)\left(\frac{D-\rho}{D}\right)\left.\left(\binom{K-\alpha}{\rho}\prod_{i=1}^{\beta-1}\binom{K-i\alpha-\rho}{\alpha}\right)^{-1}\right).	
\end{align*} It is easy to verify that~\eqref{Eq:3} and~\eqref{Eq:4} are equal for the choice of $\theta_1$ and $\theta_2$ defined as in the protocol. 

Next, consider the case of $\rho\geq D$. It is easy to see that $\mathbb{P}(j\in \boldsymbol{W}|Q)$ is given by
\begin{dmath}\label{Eq:5}
\sum_{\substack{W\subset Q_0: \\ |W|=D, j\in W}} \sum_{\substack{S\subset \mathcal{K}\setminus Q_0:\\ |S|=\alpha-\rho}}\mathbb{P}(\boldsymbol{W}=W, \boldsymbol{S}=S\cup Q_0\setminus W|Q)  
\end{dmath} for all $j\in Q_0$, and 
\begin{dmath}\label{Eq:6}
 \hspace{-0.5cm}\sum_{\substack{W\subset Q_i:\\ |W|=D, j\in W}} \mathbb{P}(\boldsymbol{W}=W, \boldsymbol{S}=Q_i\setminus W|Q)
\end{dmath} for all $j\in Q_i$, $i\in [\beta]$. Similarly as before, it can be seen that~\eqref{Eq:5} and~\eqref{Eq:6} are equal if the two quantities 
\begin{dmath}\label{Eq:7}
\sum_{\substack{W\subset Q_0: \\ |W|=D, j\in W}} \sum_{\substack{S\subset \mathcal{K}\setminus Q_0:\\ |S|=\alpha-\rho}}\mathbb{P}(Q|\boldsymbol{W}=W, \boldsymbol{S}=S\cup Q_0\setminus W)  
\end{dmath} and 
\begin{dmath}\label{Eq:8}
 \hspace{-0.5cm}\sum_{\substack{W\subset Q_i:\\ |W|=D, j\in W}} \mathbb{P}(Q|\boldsymbol{W}=W, \boldsymbol{S}=Q_i\setminus W)
\end{dmath} are equal. For any $W\subset Q_0$, $|W|=D$ such that $j\in W$, and any $S\subset \mathcal{K}\setminus Q_0$, $|S|=\alpha-\rho$, we have 
\begin{align*}
& \mathbb{P}(Q|\boldsymbol{W}=W, \boldsymbol{S}=S\cup Q_0\setminus W) = \left(\frac{\theta_1}{\theta_1+\theta_3}\right) \beta!\left(\binom{M}{\rho-D}\prod_{i=0}^{\beta-1}\binom{K-i\alpha-\rho}{\alpha}\right)^{-1}.	
\end{align*} and accordingly,~\eqref{Eq:7} is equal to 
\begin{align*}
& \left(\frac{\theta_1}{\theta_1+\theta_3}\right) \binom{\rho}{D}\binom{K-\rho}{\alpha-\rho}\beta!\left(\binom{M}{\rho-D}\prod_{i=0}^{\beta-1}\binom{K-i\alpha-\rho}{\alpha}\right)^{-1}.	
\end{align*} Fix an $i\in [\beta]$. For any $W\subset Q_i$, $|W|=D$ such that $j\in W$, we have 
\begin{align*}
& \mathbb{P}(Q|\boldsymbol{W}=W, \boldsymbol{S}=Q_i\setminus W) = \left(\frac{\theta_3}{\theta_1+\theta_3}\right) (\beta-1)!\left(\prod_{i=1}^{\beta-1}\binom{K-i\alpha}{\alpha}\right)^{-1}.	
\end{align*} and accordingly,~\eqref{Eq:8} is equal to 
\begin{align*}
& \left(\frac{\theta_3}{\theta_1+\theta_3}\right) \binom{\alpha-1}{M}(\beta-1)!\left(\prod_{i=1}^{\beta-1}\binom{K-i\alpha}{\alpha}\right)^{-1}.	
\end{align*} Again, for the choice of $\theta_1$ and $\theta_3$ as in the protocol, it is easy to verify that~\eqref{Eq:7} and~\eqref{Eq:8} are equal.
\end{proof}

\section{Proof of Theorem~\ref{thm:IPIRSI-Spc}}\label{sec:IPIRSI-Spc}
In this section, we first present a new combinatorial conjecture which, if holds, proves the tightness of the result of Theorem~\ref{thm:IPIRSI-Gen}. Next, we prove the simplest non-trivial case of this conjecture, yielding the tightness of the capacity lower bound in Theorem~\ref{thm:IPIRSI-Gen} for $M=1$ and $D=2$. 

Before stating the conjecture, we give two necessary conditions, due to individual-privacy and recoverability, which are essential to relate the IPIR-SI problem to our conjecture. 

\begin{lemma}\label{lem:2}
For any $W\in \mathcal{W}$ and $S\in \mathcal{S}$ where $S\cap W=\emptyset$, and any ${j\in \mathcal{K}}$, there must exist $W^{*}\in \mathcal{W}$, $j\in W^{*}$ and ${S^{*}\in \mathcal{S}}$ where $S^{*}\cap W^{*}=\emptyset$, such that \[H(X_{W^{*}}| A^{[W,S]}, Q^{[W,S]}, X_{S^{*}}) = 0.\] 	
\end{lemma}

\begin{proof}
The proof is by the way of contradiction, and is omitted for brevity. 	
\end{proof}

\begin{lemma}\label{lem:3}
For any $W\in \mathcal{W}$ and $S\in \mathcal{S}$ where $S\cap W=\emptyset$, and any $J\subseteq \mathcal{K}$, if $\mathbb{P}(\cup_{j\in J} E_j|Q^{[W,S]})=1$, then $|J|\geq \frac{K}{D}$, where $E_j$ for $j\in J$ is the event that $j\in \boldsymbol{W}$. 
\end{lemma}

\begin{proof}
Take any $J\subseteq \mathcal{K}$ such that $\mathbb{P}(\cup_{j\in J} E_j|Q^{[W,S]})=1$. By the union bound, $\mathbb{P}(\cup_{j\in J} E_j|Q^{[W,S]})$ is bounded from above by $\sum_{j\in J} \mathbb{P} (E_j|Q^{[W,S]})$, or equivalently, $\frac{|J|D}{K}$, noting that $\mathbb{P}(E_j|Q^{[W,S]})=\frac{D}{K}$ for all $j\in \mathcal{K}$ (by the individual-privacy condition). Since $\frac{|J|D}{K}\geq 1$, then $|J|\geq \lceil \frac{K}{D}\rceil$. 
\end{proof}

% Let $K\geq 1$, $M\geq 1$, and $D\geq 2$ be arbitrary integers such that $D+M\leq K$. Let $\mathcal{K}\triangleq [K]$.

% (see Theorem~\ref{thm:IPIRSI-Gen})

We would like to show that $H(A^{[W,S]})$, or particularly $H(A^{[W,S]}|Q^{[W,S]})$, for any protocol $(Q^{[W,S]},A^{[W,S]})$ that satisfies the conditions in Lemmas~\ref{lem:2} and~\ref{lem:3}, is bounded from below by ${\min\{K-M\lfloor \frac{K}{M+D}\rfloor, D\lceil\frac{K}{M+D}\rceil\}}$. Any such a protocol can be represented by an oracle as follows.  

% (Think of the $l$-subsets (for $0\leq l\leq M$) in $\mathcal{I}$ as the potential side information index sets $S^{*}$ and the $D$-subsets in $\mathcal{J}$ as the possible demand index sets $W^{*}$.) 

 Let $K\geq 1$, $M\geq 1$, and $D\geq 2$ be arbitrary integers such that $D+M\leq K$. Let $\mathcal{I}$ and $\mathcal{J}$ be the set of all subsets $I$ and $J$ of $\mathcal{K}\triangleq [K]$ such that $0\leq |I|\leq M$ and $|J|\geq D$, respectively. Let $f: \mathcal{I}\rightarrow \mathcal{J}$ be an arbitrary set relation (mapping). A relation $f$ is called \emph{good} if the following conditions hold: 
 \begin{itemize}
 \item[(i)] 	$I\subseteq f(I)$ for any $I\in \mathcal{I}$; 
 \item[(ii)] For any $j\in \mathcal{K}$, there exist $I\in \mathcal{I}$ and $J\in \mathcal{J}$, $|J|=D$, $j\in J$ where $I\cap J=\emptyset$ such that $J\subseteq f(I)$; 
 \item[(iii)] For any $I_1,I_2\in \mathcal{I}$, if $I_2\subseteq f(I_1)$, then $f(I_2)\subseteq f(I_1)$;
 \item[(iv)] For any $J^{*}\subseteq \mathcal{K}$, $|J^{*}|< \lceil \frac{K}{D}\rceil$ there exists (non-empty) $I\in \mathcal{I}$ such that $f(I)\cap J^{*}= \emptyset$. 
 \end{itemize}

Thinking of the $l$-subsets (for $0\leq l\leq M$) in $\mathcal{I}$ as the potential side information index sets $S^{*}$ and the $D$-subsets in $\mathcal{J}$ as the possible demand index sets $W^{*}$, one can observe that a good relation $f$, satisfying the conditions (i)-(iv), represents an arbitrary protocol that satisfies the conditions in Lemmas~\ref{lem:2} and~\ref{lem:3}. Then, it holds that for any IPIR-SI protocol, $H(A^{[W,S]}|Q^{[W,S]})\geq K-\theta$ (for any integer $\theta\geq 0$) so long as for any good relation $f$ (defined earlier) there exists a subset $I^{*}\subseteq \mathcal{K}$ of size at most $\theta$ such that the union of $f(I)$ for all $I\subseteq I^{*}$ is equal to $\mathcal{K}$. This is because, thinking of $f$ (or in turn, $(Q^{[W,S]}, A^{[W,S]})$) as an oracle, given the messages $\{X_j\}_{j\in I^{*}}$, all other messages $\{X_j\}_{j\in \mathcal{K}\setminus I^{*}}$ are recoverable from $A^{[W,S]}$ and $Q^{[W,S]}$; and hence, $H(A^{[W,S]}|Q^{[W,S]})\geq K-|I^{*}|\geq K-\theta$, as desired. 

% (for subset $I^{*}$ of size at most $\theta$) 

%Also, the condition (iv) satisfies the condition in Lemma~\ref{lem:3}, which implies that there is no collection of strictly less than $\lceil \frac{K}{D}\rceil$ indices in $\mathcal{K}$ such that some subset of which are potential demand indices for any possible side information index set. 

% : I\in \mathcal{I}

\begin{conjecture}\label{conj:1}
For any good relation $f$, there exists $I^{*}\subset \mathcal{K}$, $|I^{*}|\leq \max\{K-D\lceil \frac{K}{M+D}\rceil, M \lfloor \frac{K}{M+D}\rfloor\}$ such that $\cup_{I\subseteq I^{*}} f(I)=\mathcal{K}$.
\end{conjecture}

For $M=1$ and $D\geq 2$, the statement of Conjecture~\ref{conj:1} can be rephrased in the language of graph theory as follows. Let $G = (V,E)$ be an arbitrary directed graph (without parallel edges), where $V$ and $E$ are the set of nodes and edges of $G$, respectively. Denote by $d_{\mathrm{in}}(v)$ and $d_{\mathrm{out}}(v)$ the in-degree and out-degree of node $v\in V$, respectively, over $G$. We define an \emph{external} (or respectively, \emph{internal}) \emph{mother vertex-set} of $G$ as a minimal subset $I^{*}$ of nodes in $V$ from which all other nodes $u$ in $V\setminus I^{*}$ such that $d_{\mathrm{out}}(u)\neq 0$ (or respectively, $d_{\mathrm{in}}(u)\neq 0$) can be reached (i.e., for any $u\in V\setminus I^{*}$, $d_{\mathrm{out}}(u)\neq 0$ (or respectively, $d_{\mathrm{in}}(u)\neq 0$), there exists $v\in I^{*}$ such that there is a directed path from $v$ to $u$ in $G$), and denote the size of an external (or respectively, internal) mother vertex-set $I^{*}$ of $G$ by $\mu_{\mathrm{ext}}(G)$ (or respectively, $\mu_{\mathrm{int}}(G)$). Also, let $G^{\mathsf{T}}$ be the transpose of $G$, which is formed by reversing the direction of all edges in $G$ (i.e., $G^{\mathsf{T}}=(V,E^{\mathsf{T}})$, where $E^{\mathsf{T}} = \{(u,v): (v,u)\in E\}$). We call $G$ a \emph{$D$-graph} if the following conditions hold: 
\begin{itemize}
\item[(i)] For any $v\in V$, $d_{\mathrm{in}}(v)\geq 1$; 
\item[(ii)] For any $v\in V$, either $d_{\mathrm{out}}(v)=0$, or $d_{\mathrm{out}}(v)\geq D$; 
\item[(iii)] $\mu_{\mathrm{int}}(G^{\mathsf{T}})\geq \lceil\frac{K}{D}\rceil$. 
\end{itemize}

\begin{conjecture}\label{conj:2}
For any $D$-graph $G$ on $K$ nodes, ${\mu_{\mathrm{ext}}(G)\leq \lfloor \frac{K}{D+1}\rfloor}$. 	
\end{conjecture}

% Also, one can see that if $T(G^{\mathsf{T}})\geq \lceil \frac{K}{D}\rceil$, then $f$ satisfies the condition (iv) for a good relation. 

% Thus, $f$ is indeed a good relation. 

% (via a directed path)

Note that Conjecture~\ref{conj:2} is equivalent to Conjecture~\ref{conj:1} for $M=1$. (Since $K-D\lceil \frac{K}{D+1}\rceil\leq \lfloor \frac{K}{D+1}\rfloor$ for any $D\leq K-1$, then the upper bound on $|I^{*}|$ in Conjecture~\ref{conj:1} for $M=1$ reduces to $\lfloor\frac{K}{D+1}\rfloor$.) For any $D$-graph $G = (V,E)$ on $K$ nodes, we can define $f(v)$ for any $v\in V$ as the set of all nodes (including $v$) that can be reached from node $v$ (via a directed path in $G$). Then, it is easy to verify that $f$ satisfies the conditions (i)-(iv) for a good relation. Note also that $\mu_{\mathrm{ext}}(G)$ represents the size of a (minimal) subset $I^{*}\subseteq V$ such that $\cup_{v\in I^{*}} f(v) = V$. This indeed shows the equivalence between the two conjectures for $M=1$.  

In the following, we prove Conjecture~\ref{conj:2} for $M=1$ and $D=2$, and hence the proof of Theorem~\ref{thm:IPIRSI-Spc}.  

% if ${T(G^{\mathsf{T}})\geq \lceil \frac{K}{2}\rceil}$, then 

\begin{lemma}\label{lem:4}
For any $2$-graph $G$ on $K$ nodes, $\mu_{\mathrm{ext}}(G)\leq \lfloor \frac{K}{3}\rfloor$. 	
\end{lemma}

\begin{proof}
Let $G$ be an arbitrary $2$-graph on $K$ nodes. Suppose that $\mu_{\mathrm{ext}}(G)>\lfloor \frac{K}{3}\rfloor$. We need to show a contradiction. Let $n\triangleq \mu_{\mathrm{ext}}(G)$. Consider an arbitrary partition of the nodes in $G$ into $n$ parts, $V_1,\dots,V_n$, such that each part $V_j$ contains a node $v_j$ from which all other nodes in $V_j$ can be reached. (Note that a node in a part can potentially reach some other nodes in other parts.) Obviously, $I^{*}\triangleq \{v_1,\dots,v_n\}$ is an external mother vertex-set of $G$. 

By the minimality of $I^{*}$, it follows that no node $v_j$ can be reached from any node out of the part $V_j$. (Otherwise, from the nodes in $I^{*}\setminus \{v_j\}$ all other nodes can be reached, and this contradicts the minimality of $I^{*}$.) Since $d_{\mathrm{in}}(v_j)\geq 1$ (by definition), then there must exist another node $u_j$ in $V_j$ that reaches $v_j$. Also, no part $V_j$ can contain only a single node $v_j$, simply because $d_{\mathrm{in}}(v_j)\geq 1$, and the node $v_j$ can be reached from some other node(s) in some other part(s), which again contradicts the minimality of $I^{*}$. 

Take an arbitrary part $V_j = \{v_j,u_j\}$ of size $2$ (if exists). Since $v_j$ reaches $u_j$ (over $G$), then ${d_{\mathrm{out}}(v_j)\geq 1}$, and particularly, ${d_{\mathrm{out}}(v_j)\geq 2}$ (noting that $G$ is a $2$-graph). Thus, the node $v_j$ reaches some other node(s), say $w$, in some other part(s) over $G$. Equivalently, the node $w$ reaches both nodes $v_j$ and $u_j$ over $G^{\mathsf{T}}$. For any other part $V_j$ of size $i\geq 3$, the nodes $v_j$ and $u_j$ can be reached from each node in $V_j\setminus \{v_j,u_j\}$ over $G^{\mathsf{T}}$. 

Putting these arguments together, it follows that each node in $\{v_j,u_j\}_{j\in [n]}$ can be reached from some node(s) in $J^{*}\triangleq V\setminus \{v_j,u_j\}_{j\in [n]}$ via a directed path in $G^{\mathsf{T}}$. Then, $\mu_{\mathrm{int}}(G^{\mathsf{T}})\leq |J^{*}|= K-2n$. By assumption, $\mu_{\mathrm{ext}}(G)=n>\lfloor\frac{K}{3}\rfloor$. Thus, $|J^{*}|<K-2\lfloor \frac{K}{3}\rfloor$, and consequently, $\mu_{\mathrm{int}}(G^{\mathsf{T}}) <K-2\lfloor \frac{K}{3}\rfloor$. Since ${K-2\lfloor \frac{K}{3}\rfloor\leq \lceil\frac{K}{2}\rceil}$, then $\mu_{\mathrm{int}}(G^{\mathsf{T}}) <\lceil\frac{K}{2}\rceil$. This is a contradiction because $\mu_{\mathrm{int}}(G^{\mathsf{T}})\geq \lceil \frac{K}{2}\rceil$ for any $2$-graph $G$ on $K$ nodes. 
\end{proof}

\bibliographystyle{IEEEtran}
\bibliography{PIR_salim,pir_bib,coding1,coding2}

\end{document}